\documentclass[12pt]{article}
\usepackage[height=220mm,width=150mm]{geometry}
\usepackage{array}
\usepackage{color}
\usepackage{mathrsfs}
\usepackage{hyperref}
\usepackage{amssymb}
\usepackage{amsmath}
\usepackage{amsthm}
\usepackage{color}
\allowdisplaybreaks
\usepackage{graphicx}
\usepackage{subfigure}
\usepackage{tikz}
\usetikzlibrary{shapes.arrows,chains,positioning}
\newtheorem{Theorem}{Theorem}[section]
\newtheorem{Lemma}[Theorem]{Lemma}

\newtheorem{Proposition}[Theorem]{Proposition}
\newtheorem{Corollary}[Theorem]{Corollary}

\def\nocolor#1{}

\begin{document}

\title{Uncertainty Principles for the Offset Linear Canonical Transform}

\author{Haiye Huo\footnote{Email: hyhuo@ncu.edu.cn}\\
Department of Mathematics, School of Science, Nanchang University,\\ Nanchang~330031, Jiangxi, China \\
}

\date{}
\maketitle

\textit{Abstract}.\,\,
The offset linear canonical transform (OLCT) provides a more general framework for a number of well known linear integral transforms in signal processing and optics, such as Fourier transform, fractional Fourier transform, linear canonical transform. In this paper, to characterize simultaneous localization of a signal and its OLCT, we extend some different uncertainty principles (UPs), including Nazarov's UP, Hardy's UP, Beurling's UP, logarithmic UP and entropic UP, which have already been well studied in the Fourier transform domain over the last few decades, to the OLCT domain in a broader sense.


\section{Introduction}\label{sec:I}

Uncertainty principle (UP) plays an important role in quantum mechanics \cite{Heisenberg1927} and signal processing \cite{Benedicks1985}. In quantum mechanics, UP was first proposed by the German physicist W. Heisenberg in 1927 \cite{Heisenberg1927}. It basically says that the more precisely the position of a particle is determined, the less precisely its momentum can be known, and vice versa. From the perspective of signal processing, UP can be described as follows: ``One cannot sharply localize a signal in both the time domain and frequency domain simultaneously" (see \cite{Benedicks1985,Grochenig2001} for more details). By using different notations of essential support, there are many different kinds of UPs associated with the Fourier transform, like Heisenberg's UP \cite{HF2017, Heisenberg1927}, Nazarov's UP \cite{Jaming2007, Nazarov1993}, Hardy's UP \cite{HF2017,Hardy1933}, Beurling's UP \cite{Hogan2007}, logarithmic UP \cite{Beckner1995}, entropic UP \cite{FGS1997}, and so on.

It is well known that the offset linear canonical transform (OLCT) \cite{PD2003,Stern2007,XQ2014,XCH2015,ZWZ2016} is a generalized version of Fourier transform and has wide applications in signal processing and optics. Note that UP cannot be avoided and owns its specific form for each time-frequency representation. Therefore, it is necessary to extend the aforementioned UPs to the OLCT domain. In 2007, A. Stern extended Heisenberg's UP from the Fourier transform domain to the OLCT domain \cite{Stern2007}. It states that a nonzero function and its OLCT cannot both be sharply localized. To the best of our knowledge, there are no other results published about UPs associated with the OLCT. Hence, in this paper, the other five UPs for the Fourier transform are extended to the OLCT domain, i.e., Nazarov's UP, Hardy's UP, Beurling's UP, logarithmic UP, and entropic UP.

The rest of the paper is organized as follows. In Section \ref{sec:P}, we recall the notations of the OLCT and the generalized Parseval formula for the OLCT. In Section \ref{sec:U}, we extend the corresponding results of Nazarov's UP, Hardy's UP, Beurling's UP, logarithmic UP, and entropic UP, to the OLCT domain, respectively. In Section \ref{sec:C}, we conclude the paper.

\section{Preliminaries}\label{sec:P}

In this section, let us review the definition of the OLCT and its generalized Parseval formula.

For a given function $f(t)\in L^2(\mathbb{R})$, the definition of its OLCT \cite{KMZ2013} with parameter
$A=\left[
      \begin{array}{cc|c}
        a & b & \tau \\
        c & d & \eta \\
      \end{array}
      \right]$
is
\begin{equation}\label{eq:OLCT}
O_{A}f(u)=O_{A}[f(t)](u)=
    \begin{cases}
    \int_{-\infty}^{+\infty}f(t)K_{A}(t,u){\rm{d}}t, & b\ne0,\\
     \sqrt{d}e^{j\frac{cd}{2}(u-\tau)^2+ju\eta}f(d(u-\tau)), & b=0,\\
     \end{cases}
\end{equation}
where
\[
K_{A}(t,u)=\frac{1}{\sqrt{j2\pi b}}e^{j\frac{a}{2b}t^2-j\frac{1}{b}t(u-\tau)-j\frac{1}{b}u(d\tau-b\eta)+j\frac{d}{2b}(u^2+\tau^2)},
\]
parameters $a,\;b,\;c,\;d,\;\tau,\;\eta\in\mathbb{R}$, and $ad-bc=1$.

From the definition of the OLCT, when $b=0$, the OLCT reduces to a chirp multiplication operator. Hence, without loss of generality, we assume $b>0$ throughout the paper.

By (\ref{eq:OLCT}), one can easily check that the OLCT includes many well-known linear transforms as special cases. For instance, let $A=\left[
      \begin{array}{cc|c}
        0 & 1 & 0 \\
        -1 & 0 & 0 \\
      \end{array}
      \right]$,
the OLCT reduces to the Fourier transform \cite{Grochenig2001}; let
$A=\left[
      \begin{array}{cc|c}
        \cos\alpha & \sin\alpha & 0 \\
        -\sin\alpha & \cos\alpha & 0 \\
      \end{array}
      \right],$
the OLCT reduces to the fractional Fourier transform \cite{TZW2011}; let
$A=\left[
      \begin{array}{cc|c}
        a & b & 0 \\
        c & d & 0 \\
      \end{array}
      \right],$
the OLCT reduces to the linear canonical transform \cite{HS2015,SLZ2014,WRL2012,XS2013}, etc.

Next, we introduce one of important properties for the OLCT, i.e., its generalized Parseval formula \cite{BZ2017}, as follows:
\begin{equation}\label{OLCT:pro3}
  \int_{\mathbb{R}}f(t)\overline{g(t)}\textrm{d}t=\int_{\mathbb{R}}O_{A}f(u)\overline{O_{A}g(u)}\textrm{d}u,
  \end{equation}
where $\bar{\cdot}$ denotes the complex conjugate. This equation will be used in the following sections.

\section{Uncertainty Principles in the OLCT Domain}\label{sec:U}

Recall that there are many different forms of UPs in the Fourier transform domain, such as Heisenberg's UP, Nazarov's UP, Hardy's UP, Beurling's UP, logarithmic UP, entropic UP, and so on, in terms of different notations of ``localization". As far as we know, in 2009, G. Xu et al \cite{XWX2009New} extended the logarithmic UP and entropic UP to the linear canonical transform domain. Recently, Q. Zhang \cite{Zhang2016} extended the other three UPs: Nazarov's UP, Hardy's UP, Beurling's UP to the linear canonical transform domain. Considering that the OLCT is a generalized version of the Fourier transform or the linear canonical transform, it is natural and interesting to study the simultaneous localization of a function and its OLCT by further extending the aforementioned UPs to the OLCT domain. So far, there exists only one research work on Heisenberg's UP in the OLCT domain. Therefore, in this section, we investigate the other five different forms of UPs associated with the function $f$ and its OLCT $O_Af$, i.e., Nazarov's UP, Hardy's UP, Beurling's UP, logarithmic UP, and entropic UP, for the OLCT.

\subsection{Nazarov's UP}

As for Heisenberg's UP, its localization is measured by smallness of dispersions. By considering another criterion of localization, i.e., smallness of support, Nazarov's UP was first proposed by F.L. Nazarov in 1993 \cite{Nazarov1993}. It argues what happens if a nonzero function and its Fourier transform are only small outside a compact set? Let us recall the concept of Nazarov's UP for the Fourier transform \cite{Jaming2007,Nazarov1993} as follows.

\begin{Proposition}[\cite{Jaming2007,Nazarov1993}]\label{pro:N}
Let $f\in L^2(\mathbb{R})$, and $T,\;\Omega$ be two subsets of $\mathbb{R}$ with finite measure.
Then, there exists a constant $C>0$, such that
\begin{equation}\label{eq:N1}
\int_{\mathbb{R}}|f(t)|^2{\rm d}t\le Ce^{C|T||\Omega|}\Big(\int_{\mathbb{R}\backslash T}|f(t)|^2{\rm d}t
    +\int_{\mathbb{R}\backslash \Omega}|\mathcal{F}f(u)|^2{\rm d}u\Big),
\end{equation}
where $\mathcal{F}$ is the Fourier transform defined by
\[\mathcal{F}f(u)=\frac{1}{\sqrt{2\pi}}\int_{-\infty}^{+\infty}f(t)e^{-jut}{\rm{d}}t,
\]
and $|T|$ is denoted as the Lebesgue measure of $T$.
\end{Proposition}

Motivated by Proposition~\ref{pro:N}, we next extend the Nazarov's UP to the OLCT domain.

\begin{Theorem}\label{Thm:Na}
Let $f\in L^2(\mathbb{R})$, and $T,\;\Omega$ be two subsets of $\mathbb{R}$ with finite measure. Then, there exists a constant $C>0$, such that
\begin{equation}\label{eq:N1}
\int_{\mathbb{R}}|f(t)|^2{\rm d}t\le Ce^{C|T||\Omega|}\Big(\int_{\mathbb{R}\backslash T}|f(t)|^2{\rm d}t
    +\int_{\mathbb{R}\backslash (\Omega b)}|O_{A}f(u)|^2{\rm d}u\Big).
\end{equation}
\end{Theorem}

\begin{proof}
By the definition of the OLCT (\ref{eq:OLCT}), we can rewrite the OLCT as follows
\begin{equation}\label{Thm:Na1}
O_{A}f(u)=\frac{1}{\sqrt{jb}}e^{-j\frac{1}{b}u(d\tau-b\eta)+j\frac{d}{2b}(u^2+\tau^2)}G(u),
\end{equation}
where
\begin{equation}\label{Thm:Na2}
G(u)\triangleq\frac{1}{\sqrt{2\pi}}\int_{-\infty}^{+\infty}f(t)e^{j\frac{a}{2b}t^2-j\frac{1}{b}t(u-\tau)}{\rm{d}}t.
\end{equation}
Thus,
\begin{equation}\label{Thm:Na3}
|O_{A}f(u)|=\frac{1}{\sqrt{b}}|G(u)|.
\end{equation}
Let $g(t)=f(t)e^{j\frac{a}{2b}t^2+j\frac{1}{b}t\tau}$, then $G(ub)$ is the Fourier transform of $g(t)$, and $|f(t)|=|g(t)|$.
Since $f\in L^2(\mathbb{R})$, then $g\in L^2(\mathbb{R})$. Applying Proposition~\ref{pro:N} with the function $g(t)$ and its Fourier transform $G(ub)$, we get
\begin{equation}\label{Thm:Na4}
\int_{\mathbb{R}}|g(t)|^2{\rm d}t\le Ce^{C|T||\Omega|}\Big(\int_{\mathbb{R}\backslash T}|g(t)|^2{\rm d}t
    +\int_{\mathbb{R}\backslash \Omega}|G(ub)|^2{\rm d}u\Big).
\end{equation}
Substituting (\ref{Thm:Na3}) into (\ref{Thm:Na4}), we obtain
\begin{eqnarray*}
\int_{\mathbb{R}}|f(t)|^2{\rm d}t
&=&\int_{\mathbb{R}}|g(t)|^2{\rm d}t\\
&\le& Ce^{C|T||\Omega|}\Big(\int_{\mathbb{R}\backslash T}|f(t)|^2{\rm d}t
    +\int_{\mathbb{R}\backslash \Omega}|G(ub)|^2{\rm d}u\Big)\\
&=& Ce^{C|T||\Omega|}\Big(\int_{\mathbb{R}\backslash T}|f(t)|^2{\rm d}t
    +\frac{1}{b}\int_{\mathbb{R}\backslash (\Omega b)}|G(u)|^2{\rm d}u\Big)\\
&=& Ce^{C|T||\Omega|}\Big(\int_{\mathbb{R}\backslash T}|f(t)|^2{\rm d}t
    +\frac{1}{b}\int_{\mathbb{R}\backslash (\Omega b)}\big(\sqrt{b}|O_{A}f(u)|\big)^2{\rm d}u\Big)\\
&=& Ce^{C|T||\Omega|}\Big(\int_{\mathbb{R}\backslash T}|f(t)|^2{\rm d}t
    +\int_{\mathbb{R}\backslash (\Omega b)}|O_{A}f(u)|^2{\rm d}u\Big),
\end{eqnarray*}
which completes the proof.
\end{proof}

Theorem~\ref{Thm:Na} is a quantitative version of UP, the constant $C$ on the right-hand side of (\ref{eq:N1}) is hard to determine exactly. By Theorem~\ref{Thm:Na}, we know that it is not possible for a nonzero function $f$ and its OLCT $O_Af$ to both be supported on sets of finite Lebsgue measure.

\subsection{Hardy's UP}

Hardy's UP was first introduced by G.H. Hardy in 1933 \cite{Hardy1933}. Its localization is measured by fast decrease of a function and its Fourier transform. Hardy's UP basically says that it is impossible for a nonzero function and its Fourier transform to decrease very rapidly simultaneously. Let us review Hardy's UP in the Fourier transform domain \cite{Hardy1933,Hogan2007} as follows.

\begin{Proposition}[\cite{Hogan2007}, Theorem~5.2.1]\label{pro:H}
If a function $f\in L^2(\mathbb{R})$ is such that
\[|f(t)|=\mathcal{O}(e^{-\pi\alpha t^2})
\]
and
\[
|\mathcal{F}f(u)|=\mathcal{O}\big(e^{-u^2/(4\pi\alpha)}\big)
\]
for some positive constant $\alpha>0$, then
\begin{equation}\label{eq:H}
f(t)=Ce^{-\pi\alpha t^2}
\end{equation}
for some $C\in \mathbb{C}$.
\end{Proposition}

Based on Proposition~{\ref{pro:H}}, we derive the corresponding Hardy's UP for the OLCT.

\begin{Theorem}\label{thm:H1}
If a function $f\in L^2(\mathbb{R})$ is such that
\[|f(t)|=\mathcal{O}(e^{-\pi\alpha t^2})
\]
and
\[
|O_A f(u)|=\mathcal{O}\big(e^{-u^2/(4\pi\alpha b^2)}\big)
\]
for some positive constant $\alpha>0$, then
\begin{equation}\label{eq:H1}
f(t)=Ce^{-\big(\pi\alpha+j\frac{a}{2b}\big)t^2-j\frac{1}{b}t\tau}
\end{equation}
for some $C\in \mathbb{C}$. Here $\tau$ is a parameter of $A$.
\end{Theorem}

\begin{proof}
Let $g(t)=f(t)e^{j\frac{a}{2b}t^2+j\frac{1}{b}t\tau}$. Since $|f(t)|=\mathcal{O}(e^{-\pi\alpha t^2})$, we have $|g(t)|=|f(t)|=\mathcal{O}(e^{-\pi\alpha t^2})$. Let the OLCT $O_{A}f(u)$ be rewritten as (\ref{Thm:Na1}), and $G(u)$ be given by (\ref{Thm:Na2}). Thus,
\begin{eqnarray*}
|G(u)|
&=&\sqrt{b}|O_{A}f(u)|\\
&=&\mathcal{O}\big(e^{-u^2/(4\pi\alpha b^2)}\big).
\end{eqnarray*}
Therefore, we have
\begin{eqnarray*}
|G(ub)|
&=&\mathcal{O}\big(e^{-(ub)^2/(4\pi\alpha b^2)}\big)\\
&=&\mathcal{O}\big(e^{-u^2/(4\pi\alpha)}\big).
\end{eqnarray*}
Since $G(ub)$ is the Fourier transform of $g(t)$, it follows from Proposition~\ref{pro:H} that
\[
g(t)=Ce^{-\pi\alpha t^2}
\]
for some $C\in \mathbb{C}$.
Hence, we obtain
\begin{eqnarray*}
f(t)
&=&g(t)e^{-j\frac{a}{2b}t^2-j\frac{1}{b}t\tau}\\
&=&Ce^{-\big(\pi\alpha+j\frac{a}{2b}\big)t^2-j\frac{1}{b}t\tau}.
\end{eqnarray*}
This completes the proof.
\end{proof}

It follows from Theorem~\ref{thm:H1} that it is impossible for a nonzero function $f$ and its OLCT $O_Af$ both to decrease very rapidly.

\subsection{Beurling's UP}

Beurling's UP is a variant of Hardy's UP. It implies the weak form of Hardy's UP immediately. Let us revisit Beurling's UP in the Fourier transform domain \cite{Hogan2007} as follows.

\begin{Proposition}[\cite{Hogan2007}]\label{pro:B}
Let $f\in L^1(\mathbb{R})$ and $\mathcal{F}f\in L^1(\mathbb{R})$. If
\begin{equation}\label{eq:B}
\int_{\mathbb{R}^2}|f(t)\mathcal{F}f(u)|e^{|tu|}{\rm d}t{\rm d}u<\infty,
\end{equation}
then $f=0$.
\end{Proposition}

Next, we formulate the Beurling's UP in the OLCT domain.

\begin{Theorem}\label{thm:B}
Let $f\in L^1(\mathbb{R})$ and $O_{A}f\in L^1(\mathbb{R})$. If
\begin{equation}\label{thm:B1}
\int_{\mathbb{R}^2}|f(t)O_{A}f(u)|e^{|tu/b|}{\rm d}t{\rm d}u<\infty,
\end{equation}
then $f=0$.
\end{Theorem}

\begin{proof}
Let $g(t)=f(t)e^{j\frac{a}{2b}t^2+j\frac{1}{b}t\tau}$, the OLCT $O_{A}f(u)$ be rewritten in the form of (\ref{Thm:Na1}), and $G(u)$ be defined by (\ref{Thm:Na2}). Since $f(t)\in L^1(\mathbb{R})$ and $O_{A}f(u)\in L^1(\mathbb{R})$, we get $g(t)\in L^1(\mathbb{R})$ and $G(u)\in L^1(\mathbb{R})$. Thus, $G(ub)\in L^1(\mathbb{R})$.
By (\ref{thm:B1}), we obtain
\begin{eqnarray*}
\int_{\mathbb{R}^2}|g(t)G(ub)|e^{|tu|}{\rm d}t{\rm d}u
&=&\frac{1}{b}\int_{\mathbb{R}^2}|g(t)G(u)|e^{|tu/b|}{\rm d}t{\rm d}u\\
&=&\frac{1}{\sqrt{b}}\int_{\mathbb{R}^2}|f(t)O_Af(u)|e^{|tu/b|}{\rm d}t{\rm d}u\\
&<&\infty.
\end{eqnarray*}
Hence, it follows from Proposition~\ref{pro:B} that $g=0$. Therefore, we have $f=0$.
\end{proof}

From Theorem~\ref{thm:B}, we know that it is not possible for a nonzero function $f$ and its OLCT $O_Af$ to decrease very rapidly simultaneously.

\subsection{Logarithmic UP}

Logarithmic UP was first introduced by W. Beckner in 1995 \cite{Beckner1995}. Its localization is measured in terms of entropy. It is derived by using Pitt's inequality. First, let us recall Pitt's inequality as follows.

\begin{Lemma}[\cite{Beckner1995}]\label{lem:P}
For $f\in \mathcal{S}(\mathbb{R})$ and $0\le\lambda<1$, we have
\begin{equation}\label{eq:P}
\int_{\mathbb{R}}|u|^{-\lambda}|\mathcal{F}f(u)|^2{\rm d}u
\le\frac{\Gamma^2(\frac{1-\lambda}{4})}{\Gamma^2(\frac{1+\lambda}{4})}\int_{\mathbb{R}}|t|^{\lambda}|f(t)|^2{\rm d}t,
\end{equation}
where $\mathcal{S}(\mathbb{R})$ denotes the Schwartz class, and $\Gamma(\cdot)$ is the Gamma function.
\end{Lemma}

In order to obtain logarithmic UP associated with the OLCT, we derive the corresponding generalized Pitt's inequality for the OLCT as follows.

\begin{Theorem}\label{Thm:P}
For $f\in \mathcal{S}(\mathbb{R})$ and $0\le\lambda<1$, we have
\begin{equation}\label{Thm:P1}
b^{\lambda}\int_{\mathbb{R}}|u|^{-\lambda}|O_Af(u)|^2{\rm d}u
\le\frac{\Gamma^2(\frac{1-\lambda}{4})}{\Gamma^2(\frac{1+\lambda}{4})}\int_{\mathbb{R}}|t|^{\lambda}|f(t)|^2{\rm d}t.
\end{equation}
\end{Theorem}

\begin{proof}
Let $g(t)=f(t)e^{j\frac{a}{2b}t^2+j\frac{1}{b}t\tau}$, the OLCT $O_{A}f(u)$ be rewritten as (\ref{Thm:Na1}), and $G(u)$ be denoted as (\ref{Thm:Na2}). Since $G(ub)$ is the Fourier transform of $g(t)$, by applying Lemma~\ref{lem:P}, we obtain
\begin{equation}\label{Thm:P2}
\int_{\mathbb{R}}|u|^{-\lambda}|G(ub)|^2{\rm d}u
\le\frac{\Gamma^2(\frac{1-\lambda}{4})}{\Gamma^2(\frac{1+\lambda}{4})}\int_{\mathbb{R}}|t|^{\lambda}|g(t)|^2{\rm d}t.
\end{equation}
Let $u^{\prime}=ub$ in (\ref{Thm:P2}), we have
\begin{equation}\label{Thm:P3}
\frac{1}{b}\int_{\mathbb{R}}\Big|\frac{u^{\prime}}{b}\Big|^{-\lambda}|G(u^{\prime})|^2{\rm d}u^{\prime}
\le\frac{\Gamma^2(\frac{1-\lambda}{4})}{\Gamma^2(\frac{1+\lambda}{4})}\int_{\mathbb{R}}|t|^{\lambda}|g(t)|^2{\rm d}t.
\end{equation}
Substituting $|g(t)|=|f(t)|$ and $G(u)=\sqrt{b}|O_Af(u)|$ into (\ref{Thm:P3}), we get
\begin{equation}\label{Thm:P4}
\frac{1}{b}\int_{\mathbb{R}}\Big|\frac{u}{b}\Big|^{-\lambda}|\sqrt{b}O_Af(u)|^2{\rm d}u
\le\frac{\Gamma^2(\frac{1-\lambda}{4})}{\Gamma^2(\frac{1+\lambda}{4})}\int_{\mathbb{R}}|t|^{\lambda}|f(t)|^2{\rm d}t,
\end{equation}
i.e.,
\[
b^{\lambda}\int_{\mathbb{R}}|u|^{-\lambda}|O_Af(u)|^2{\rm d}u
\le\frac{\Gamma^2(\frac{1-\lambda}{4})}{\Gamma^2(\frac{1+\lambda}{4})}\int_{\mathbb{R}}|t|^{\lambda}|f(t)|^2{\rm d}t,
\]
which completes the proof.
\end{proof}

Based on the generalized Pitt's inequality for the OLCT proposed in Theorem~\ref{Thm:P}, we investigate the logarithmic UP associated with the OLCT.

\begin{Theorem}\label{Thm:Log}
Let $f\in \mathcal{S}(\mathbb{R})$, and $\|f\|_2=1$, then
\begin{equation}\label{eq:Log1}
\int_{\mathbb{R}}|f(t)|^2\ln|t|{\rm d}t+\int_{\mathbb{R}}|O_{A}f(u)|^2\ln|u|{\rm d}u
\ge\ln b+\frac{\Gamma^{\prime}(1/4)}{\Gamma(1/4)}.
\end{equation}
\end{Theorem}

\begin{proof}
Let
\[
M_\lambda\triangleq\frac{\Gamma^2(\frac{1-\lambda}{4})}{\Gamma^2(\frac{1+\lambda}{4})},
\]
and
\[
S(\lambda)\triangleq b^{\lambda}\int_{\mathbb{R}}|u|^{-\lambda}|O_Af(u)|^2{\rm d}u-M_\lambda\int_{\mathbb{R}}|t|^{\lambda}|f(t)|^2{\rm d}t.
\]
Taking the derivative of $S(\lambda)$ about the variable $\lambda$, we have
\begin{eqnarray*}
S^{\prime}(\lambda)&=&b^{\lambda}\ln b\int_{\mathbb{R}}|u|^{-\lambda}|O_Af(u)|^2{\rm d}u
-b^{\lambda}\int_{\mathbb{R}}|u|^{-\lambda}\ln|u||O_Af(u)|^2{\rm d}u{}\\
{}&&-M_\lambda\int_{\mathbb{R}}|t|^{\lambda}\ln{|t|}|f(t)|^2{\rm d}t
-(M_\lambda)^{\prime}\int_{\mathbb{R}}|t|^{\lambda}|f(t)|^2{\rm d}t,
\end{eqnarray*}
where
\begin{eqnarray*}
(M_\lambda)^{\prime}
&=&\bigg[-\frac{1}{2}\Gamma\Big(\frac{1-\lambda}{4}\Big)\Gamma^{\prime}\Big(\frac{1-\lambda}{4}\Big)\Gamma^2\Big(\frac{1+\lambda}{4}\Big)
-\frac{1}{2}\Gamma\Big(\frac{1+\lambda}{4}\Big)\Gamma^{\prime}\Big(\frac{1+\lambda}{4}\Big){}\\
{}&&\quad\times\Gamma^2\Big(\frac{1-\lambda}{4}\Big)\bigg]/\Gamma^4\Big(\frac{1+\lambda}{4}\Big).
\end{eqnarray*}
By Theorem~\ref{Thm:P}, we know
\[
S(\lambda)\le 0 \quad \mbox{for} \quad 0\le\lambda<1.
\]
Since $S(0)=0$,
we get
\[
S^{\prime}(0+)\le 0,
\]
that is,
\begin{align}
&\int_{\mathbb{R}}\ln|u||O_Af(u)|^2{\rm d}u+\int_{\mathbb{R}}\ln{|t|}|f(t)|^2{\rm d}t\nonumber\\
&\ge\ln b\int_{\mathbb{R}}|O_Af(u)|^2{\rm d}u
+\frac{\Gamma^{\prime}(1/4)}{\Gamma(1/4)}\int_{\mathbb{R}}|f(t)|^2{\rm d}t.\label{eq:Log2}
\end{align}
By the generalized Parseval formula for the OLCT (\ref{OLCT:pro3}), we get
\begin{equation}\label{eq:Log3}
\|O_Af\|_2=\|f\|_2=1.
\end{equation}
Substituting (\ref{eq:Log3}) into (\ref{eq:Log2}), we have
\[
\int_{\mathbb{R}}|f(t)|^2\ln{|t|}{\rm d}t+\int_{\mathbb{R}}|O_Af(u)|^2\ln|u|{\rm d}u\ge\ln b+\frac{\Gamma^{\prime}(1/4)}{\Gamma(1/4)}.
\]
This completes the proof.
\end{proof}

Applying Jassen's inequality to (\ref{eq:Log1}), it is easily to show that logarithmic UP proposed in Theorem~\ref{Thm:Log} implies Heisenberg's UP derived in \cite{Stern2007}.

\subsection{Entropic UP}

Entropic UP is a fundamental tool in information theory, physical quantum, and harmonic analysis. Its localization is measured in terms of Shannon entropy.
Let $\rho$ be a probability density function on $\mathbb{R}$. The Shannon entropy of $\rho$ is denoted as
\begin{equation}\label{eq:S}
E(\rho)=-\int_{\mathbb{R}}\rho(t)\ln \rho(t){\rm d}t.
\end{equation}
In what follows, we revisit entropic UP associated with the Fourier transform \cite{FGS1997}.

\begin{Proposition}[\cite{FGS1997}]\label{pro:E}
Let $f\in L^2(\mathbb{R})$, and $\|f\|_2=1$, then
\begin{equation}\label{eq:E}
E(|f|^2)+E(|\mathcal{F}f|^2)\ge \ln(\pi e).
\end{equation}
\end{Proposition}

Next, we propose the entropic UP in the OLCT domain.

\begin{Theorem}\label{Thm:En}
Let $f\in L^2(\mathbb{R})$, and $\|f\|_2=1$, then
\begin{equation}\label{Thm:En1}
E(|f|^2)+E(|O_{A}f|^2)\ge \ln(\pi eb).
\end{equation}
\end{Theorem}

\begin{proof}
Let $g(t)=f(t)e^{j\frac{a}{2b}t^2+j\frac{1}{b}t\tau}$, the OLCT $O_{A}f(u)$ be rewritten in the form of (\ref{Thm:Na1}), and $G(u)$ be defined by (\ref{Thm:Na2}). Since $\|f\|_2=1$, we have $\|g\|_2=\|f\|_2=1$. Hence, by Proposition~\ref{pro:E}, we get
\begin{equation}\label{Thm:En2}
-\int_{\mathbb{R}}|g(t)|^2\ln|g(t)|^2{\rm d}t-\int_{\mathbb{R}}|G(ub)|^2\ln|G(ub)|^2{\rm d}u\ge\ln(\pi e).
\end{equation}
Let $u^{\prime}=ub$ in (\ref{Thm:En2}), we have
\begin{equation}\label{Thm:En3}
-\int_{\mathbb{R}}|g(t)|^2\ln|g(t)|^2{\rm d}t-\frac{1}{b}\int_{\mathbb{R}}|G(u^{\prime})|^2\ln|G(u^{\prime})|^2{\rm d}u^{\prime}\ge\ln(\pi e).
\end{equation}
Substituting $|g(t)|=|f(t)|$, and $|G(u)|=\sqrt{b}|O_Af(u)|$ into (\ref{Thm:En3}), we obtain
\[
-\int_{\mathbb{R}}|f(t)|^2\ln|f(t)|^2{\rm d}t-\int_{\mathbb{R}}|O_Af(u)|^2\ln\big(b|O_{A}f(u)|^2\big){\rm d}u
\ge\ln(\pi e).
\]
Using the fact that $\|O_Af\|_2=\|f\|_2=1$, we get
\begin{eqnarray*}
E(|f|^2)+E(|O_{A}f|^2)
&=&-\int_{\mathbb{R}}|f(t)|^2\ln|f(t)|^2{\rm d}t-\int_{\mathbb{R}}|O_Af(u)|^2\ln|O_{A}f(u)|^2{\rm d}u\\
&\ge&\ln b\int_{\mathbb{R}}|O_{A}f(u)|^2{\rm d}u+\ln(\pi e)\\
&=&\ln b+\ln(\pi e)\\
&=&\ln(\pi e b),
\end{eqnarray*}
which completes the proof.
\end{proof}

Theorem~\ref{Thm:En} measures the incompatibility of measurements in terms of Shannon entropy.
We next demonstrate that Theorem~\ref{Thm:En} can also imply the Heisenberg's UP mentioned in \cite{Stern2007}.

At the beginning, let us introduce some notations. Let $\mu$ be a probability measure on $\mathbb{R}$. The variance of $\mu$ is defined as
\begin{equation}\label{Add1}
V(\mu)=\inf_{\xi\in \mathbb{R}}\int_{\mathbb{R}}(t-\xi)^2{\rm{d}}{\mu(t)}.
\end{equation}
If the integral on the right side of $(\ref{Add1})$ is finite for one value of $\xi$, then it is finite for every $\xi$. In this case, it can achieve the minimization when $\xi$ is the mean of $\mu$:
\[
M(\mu)=\int_{\mathbb{R}}t{\rm{d}}\mu(t).
\]
For $\rho\in L^1(\mathbb{R})$, if ${\rm{d}}{\mu(t)}=\rho(t){\rm{d}}t$, we say $\rho$ is a probability density function, and use the notations $M(\rho)$ and $V(\rho)$ instead of $M(\mu)$ and $V(\mu)$, respectively.

Using the fact $\|f\|_2=\|O_Af\|_2$, we know that if $f\in L^2(\mathbb{R})$ and $\|f\|_2=1$, then $|f|^2$ and $|O_{A}f|^2$ both are probability density functions on $\mathbb{R}$.

\begin{Lemma}\cite[Theorem~5.1]{FGS1997}\label{Lem:Add}
Let $\rho$ be a probability density function on $\mathbb{R}$ with finite variance, then $E(\rho)$ is well defined and
\begin{equation}
E(\rho)\le\frac{1}{2}\ln[2\pi e V(\rho)].
\end{equation}
\end{Lemma}

Combining Theorem~\ref{Thm:En} and Lemma~\ref{Lem:Add}, we then immediately get the Heisenberg's UP \cite{Stern2007} as follows.
\begin{Corollary}\label{Coro:add}
Let $f\in L^2(\mathbb{R})$, and $\|f\|_2=1$, then
\begin{equation}\label{Coro:add1}
 V(|f|^2) V(|O_{A}f|^2)\ge\frac{b^2}{4},
\end{equation}
which implies
\[
\int_{\mathbb{R}}(t-\xi)^2|f(t)|^2{\rm{d}}t\int_{\mathbb{R}}(u-\zeta)^2|O_Af(u)|^2{\rm{d}}u\ge\frac{b^2}{4}
\]
for any $f\in L^2(\mathbb{R})$ and any $\xi,\;\zeta\in \mathbb{R}$.
\end{Corollary}

\section{Conclusion}\label{sec:C}

In this paper, five different forms of UPs associated with the OLCT are proposed. First, we derive Nazarov's UP for the OLCT, which is a quantitative version of UP. It shows that it is not possible for a nonzero function $f$ and its OLCT $O_Af$ to both be supported on sets of finite Lebesgue measure. Second, based on the decreasing property, we propose two UPs in the OLCT domain: Hardy's UP and its variant-Beurling's UP. These two UPs state that it is impossible for a nonzero function $f$ and its OLCT $O_Af$ to both decrease very rapidly. Finally, we generalize Pitt's inequality to the OLCT domain, and then obtain logarithmic UP for the OLCT. Moreover, with regard to Shannon entropy, we extend entropic UP to the OLCT domain. In the future work, we will consider these UPs for discrete signals.

\section*{Acknowledgements}
The author thanks the referees very much for carefully reading the paper and for elaborate and valuable suggestions.


\begin{thebibliography}{99}

\bibitem{Beckner1995}
W.~Beckner, Pitt's inequality and the uncertainty principle. Proc. Amer. Math. Soc. \textbf{123}(6), 1897--1905 (1995).

\bibitem{Benedicks1985}
M.~Benedicks, On {Fourier} transforms of functions supported on sets of finite {Lebesgue} measure. J. Math. Anal. Appl. \textbf{106}(1), 180--183 (1985).

\bibitem{BZ2017}
A.~Bhandari, A.I. Zayed, Shift-invariant and sampling spaces associated with the special affine {Fourier} transform. Appl. Comput. Harmon. Anal., In Press, 2017.

\bibitem{FGS1997}
G.B. Folland, A.~Sitaram, The uncertainty principle: a mathematical survey. J. Fourier Anal. Appl. \textbf{3}(3), 207--238 (1997).

\bibitem{Grochenig2001}
K.~Gr{\"{o}}chenig, Foundations of time-frequency analysis. Birkh{\"{a}}user, 2001.

\bibitem{HS2015}
H.~Huo, W.~Sun, Sampling theorems and error estimates for random signals in the linear canonical transform domain. Signal Process. \textbf{111}, 31--38 (2015).

\bibitem{HF2017}
Y.El Haoui, S. Fahlaoui, The uncertainty principle for the two-sided quaternion {Fourier} transform, Mediterr. J. Math. \textbf{14}(6), 221 (2017).

\bibitem{Hardy1933}
G.H. Hardy, A theorem concerning {Fourier} transforms. J. London Math. Soc. \textbf{8}(3), 227--231 (1933).

\bibitem{Heisenberg1927}
W.~Heisenberg, Uber den anschaulichen inhalt der quanten theoretischen kinematik und mechanik. Zeitschrift f\"{u}r Physik \textbf{43}, 172--198 (1927).

\bibitem{Hogan2007}
J.A. Hogan, Time-frequency and time-scale methods: {Adaptive} decompositions, uncertainty principles, and sampling. Springer Science \& Business Media, 2007.

\bibitem{Jaming2007}
P.~Jaming, Nazarov's uncertainty principles in higher dimension. J. Approx. Theory \textbf{149}(1), 30--41 (2007).

\bibitem{KMZ2013}
K.~Kou, J.~Morais, Y.~Zhang, Generalized prolate spheroidal wave functions for offset linear canonical transform in {Clifford} analysis. Math. Methods Appl. Sci. \textbf{36}(9), 1028--1041 (2013).

\bibitem{Nazarov1993}
F.L. Nazarov, Local estimates for exponential polynomials and their applications to inequalities of the uncertainty principle type. Algebra I
  Analiz \textbf{5}(4), 663--717 (1993).

\bibitem{PD2003}
S.-C. Pei, J.-J. Ding, Eigenfunctions of the offset {Fourier}, fractional {Fourier}, and linear canonical transforms. JOSA A \textbf{20}(3), 522--532 (2003).

\bibitem{SLZ2014}
J.~Shi, X.~Liu, N.~Zhang, On uncertainty principles for linear canonical transform of complex signals via operator methods. Signal Image Video Process. \textbf{8}(1), 85--93 (2014).

\bibitem{Stern2007}
A.~Stern, Sampling of compact signals in offset linear canonical transform domains. Signal Image Video Process. \textbf{1}(4), 359--367 (2007).

\bibitem{TZW2011}
R.~Tao, F.~Zhang, Y.~Wang, Sampling random signals in a fractional {Fourier} domain. Signal Process. \textbf{91}(6), 1394--1400 (2011).

\bibitem{WRL2012}
D.~Wei, Q.~Ran, Y.~Li, A convolution and correlation theorem for the linear canonical transform and its application. Circuits Syst. Signal Process. \textbf{31}(1), 301--312 (2012).

\bibitem{XQ2014}
Q.~Xiang, K.~Qin, Convolution, correlation, and sampling theorems for the offset linear canonical transform. Signal Image Video Process. \textbf{8}(3), 433--442 (2014).

\bibitem{XS2013}
L.~Xiao, W.~Sun, Sampling theorems for signals periodic in the linear canonical transform domain. Opt. Commun. \textbf{290}, 14--18 (2013).

\bibitem{XWX2009New}
G.~Xu, X.~Wang, X.~Xu, New inequalities and uncertainty relations on linear canonical transform revisit. EURASIP J. Advan. Signal Process. \textbf{2009}(1), 1--7 (2009).

\bibitem{XCH2015}
S.~Xu, Y.~Chai, Y.~Hu, Spectral analysis of sampled band-limited signals in the offset linear canonical transform domain. Circuits Syst. Signal Process. \textbf{34}(12), 3979-3997 (2015).

\bibitem{Zhang2016}
Q.~Zhang, Zak transform and uncertainty principles associated with the linear canonical transform. IET Signal Process. \textbf{10}(7), 791--797 (2016).

\bibitem{ZWZ2016}
X.~Zhi, D.~Wei, W.~Zhang, A generalized convolution theorem for the special affine {Fourier} transform and its application to filtering. Optik \textbf{127}(5), 2613--2616 (2016).

\end{thebibliography}
\end{document}